%% file: TR_SIC-1-08.tex
\documentclass[runningheads]{llncs}

\usepackage{stmaryrd} 
\usepackage{amssymb} 
\usepackage{hyperref}

\newcommand{\toy}{\mathcal{TOY}} 
\newcommand{\NAT}{\bbbn} 
\newcommand{\REAL}{\bbbr} 

\newcommand{\qdom}{$\mathcal{D}$} 
\newcommand{\qdomm}{\mathcal{D}} 
 \newcommand{\aqdomm}{D \setminus\{\bot\}}
\newcommand{\B}{\mathcal{B}} 
\newcommand{\U}{\mathcal{U}} 
\newcommand{\W}{\mathcal{W}} 

\newcommand{\cdom}[1]{\mathcal{C}_{#1}} 
\newcommand{\rdom}{\mathcal{R}} 

\newcommand{\qlp}[1]{QLP({#1})} 

\newcommand{\clp}[1]{CLP({#1})} 

\newcommand{\diff}{~{\Longleftrightarrow_{\mathrm{def}}}~} 
\newcommand{\union}{\bigcup} 
\newcommand{\inter}{\bigcap} 
\newcommand{\supr}{\bigsqcup} 
\newcommand{\infi}{\bigsqcap} 
\newcommand{\Dentail}{~{\succcurlyeq_{\qdomm}}~} 
\newcommand{\app}{\,\hat{}\,} 

\newcommand{\Prog}{\mathcal{P}} 
\newcommand{\UProg}{\mathcal{P_U}}
\newcommand{\WProg}{\mathcal{P_W}}
\newcommand{\Var}{\mathcal V\!ar} 
\newcommand{\War}{\mathcal W\!ar} 
\newcommand{\varset}[1]{\mathrm{var}(#1)} 
\newcommand{\warset}[1]{\mathrm{war}(#1)} 
\newcommand{\Tp}{\mathrm{T}_{\Prog}} 
\newcommand{\Mp}{\mathcal{M}_{\Prog}} 
\newcommand{\M}[1]{\mathcal{M}_{#1}} 
\newcommand{\intd}{\mathrm{Int}_{\Sigma}(\qdomm)} 
\newcommand{\sust}{\mathrm{Subst}_{\Sigma}} 
\newcommand{\sustd}{\sust(\qdomm)} 
\newcommand{\I}{\mathcal{I}} 
\newcommand{\J}{\mathcal{J}} 
\newcommand{\at}[2]{#1\,\sharp\,#2} 
\newcommand{\ats}[1]{\overline{#1}} 
\newcommand{\qgets}[1]{\gets #1 -} 
\newcommand{\sep}{~{\talloblong}~} 

\newcommand{\sld}[1]{SLD({#1})} 
\newcommand{\qhl}[1]{QHL({#1})} 
\newcommand{\qhld}{\vdash_{\mathrm{QHL}(\qdomm)}}
\newcommand{\qhldn}[1]{\vdash_{\mathrm{QHL}(\qdomm)}^{#1}}
\newcommand{\resx}{\Vdash}
\newcommand{\qres}[1]{\resx_{#1}}
\newcommand{\qresn}[2]{\resx^{#1}_{#2}}
\newcommand{\qsol}[1]{QSol_{\Prog}(#1)}
\newcommand{\qsoln}[2]{QSol_{\Prog}^{#1}(#2)}

\title{A Generic Scheme for \\ Qualified Logic Programming\thanks{Research partially supported by projects MERIT-FORMS (TIN2005-09027-C03-03) and PROMESAS-CAM(S-0505/TIC/0407)}}
\author{Mario Rodr\'{i}guez-Artalejo and Carlos A. Romero-D\'{i}az}
\authorrunning{M.~Rodr\'iguez-Artalejo and C.A.~Romero-D\'iaz}
\subtitle{Technical Report SIC-1-08}
\institute{Departamento de Sistemas Inform\'aticos y Computaci\'on\\
Universidad Complutense de Madrid, Spain\\
\email{mario@sip.ucm.es} and \email{cromdia@fdi.ucm.es}}

\begin{document}
\maketitle

\begin{abstract}
Uncertainty in Logic Programming has been investigated since about 25 years, publishing papers dealing with various approaches to semantics and different applications. This report is  intended as a first step towards the investigation of  {\em qualified computations} in  Constraint Functional Logic Programming, including 
{\em uncertain computations}  as a particular case.
We revise an early proposal, namely van Emden's {\em Quantitative Logic Programming} \cite{VE86}, and we improve it in two ways. Firstly, we generalize van Emden's $QLP$ to a generic scheme $\qlp{\qdomm}$ parameterized by any given {\em Qualification Domain} $\qdomm$, which must be a lattice satisfying certain natural axioms. We present several interesting instances for $\qdomm$, one of which corresponds to van Emden's $QLP$. Secondly, we generalize van Emden's results by  providing stronger ones, concerning both semantics and goal solving. We present \emph{Qualified $SLD$ Resolution} over $\qdomm$, a sound and strongly complete goal solving procedure for $\qlp{\qdomm}$, which is applicable to open goals  and can be efficiently implemented using $CLP$ technology over any constraint domain $\cdom{\qdomm}$ able to deal with qualification constraints  over $\qdomm$. 
We have developed a prototype implementation of some instances of the $\qlp{\qdomm}$ scheme
(including  van Emden's $QLP$) on top of the $CFLP$ system $\toy$.

\end{abstract}

\input{TR_SIC-1-08_1-Introduction}
\input{TR_SIC-1-08_2-Domains}

\input{TR_SIC-1-08_3-Language}
\input{TR_SIC-1-08_4-Solving}

\input{TR_SIC-1-08_5-Implementation}
\input{TR_SIC-1-08_6-Conclusions}

\bibliographystyle{abbrv}
\bibliography{biblio}

\newpage\appendix
\input{TR_SIC-1-08_7-Appendix}

\end{document}

%% file: TR_SIC-1-08_1-Introduction.tex
\section{Introduction}\label{Introduction}


The investigation of uncertainty in logic programming has proceeded along various lines during the last 25 years. A recent recollection by V. S. Subrahmanian \cite{Sub07} highlights some phases in the evolution of the topic from the viewpoint of a committed researcher.

Research on the field has dealt with various approaches to semantics, as well as different applications. One of the earliest approaches was {\em Quantitative Logic Programming}, $QLP$ for short. This can be traced back to a paper by Shapiro \cite{Sha83}, who proposed to use real numbers in the interval $(0,1]$ as {\em certainty factors}, as well as {\em certainty functions} for propagating certainty factors from the bodies to the heads of program clauses. Subsequently, van Emden \cite{VE86} considered $QLP$ with an  {\em attenuation factor} $f \in (0,1]$ attached to the implication of each program clause and restricted his attention to the certainty function which propagates to a clause head the certainty factor $f \times b$, where $f$ is the clause's attenuation factor and $b$ is the minimum of the certainty factors known for the body atoms. Van Emden's approach was less general than Shapiro's because of the fixed  choice of a particular certainty function, but it allowed to prove more general results on model theoretic and fixpoint semantics, similar to those previously obtained in \cite{VEK76,AVE82} for classical Logic Programming. Moreover, \cite{VE86} gave a procedure for computing the certainty of atoms in the least Herbrand model of a given program, by applying an alpha-beta heuristic to the atoms' and/or search trees. This procedure worked only for ground atoms having a finite search tree.

Following these beginnings, logic programming with uncertainty developed in various directions. Subrahmanian \cite{Sub87} proposed an alternative to \cite{VE86}, using a different lattice of numeric values (aiming at a separate representation of certainty degrees for truth an falsity) as well as clauses whose atoms were annotated with values from this lattice. Neither certainty functions nor attenuation factors were used in this approach, which was extended in \cite{Sub88} to provide goal solving procedures enjoying stronger soundness and completeness results. As a brief summary of some significant later contributions let us mention: generalized {\em annotated logic programs} \cite{KS92}, a quite general framework which will be discussed in more detail in Section \ref{Conclusions}; semantics based on {\em bilattices} of generalized truth values with both a `knowledge' order and a `truth' order \cite{Fit91}; logic programming with probabilistic semantics and applications to deductive databases \cite{NS92,NS93}; quantitative and probabilistic constraint logic programming and applications to natural language processing \cite{Rie98phd}; {\em hybrid probabilistic programs} \cite{DS00}; probabilistic agent programs \cite{DNS00} and their extension to deal with both time and uncertainty \cite{DKS06}; logic programs with similarity based unification and applications to flexible data retrieval 
\cite{AF99,Ses02,GMV04,LSS04}; and functional logic programming with similarity based unification \cite{MP06a}.


We are interested in a long-term research project aiming at a generalization of existing work on logic  programming with uncertainty.
The generalization we plan to develop will operate in two directions: 
a) extending logic programming languages to  more expressive multi-paradigm declarative languages supporting {\em functions} and {\em constraints}; 
and b) generalizing uncertain truth values to so-called {\em qualification values}, attached to computed answers and intended to  
measure the degree in which such computed answers satisfy various user's expectations.
In this setting, (constraint) logic programming with uncertainty becomes the particular case in which no  functional programming features are used and qualification values are just uncertain truth values. 
As a first step, we present in this report a generalization  of the early $QLP$ proposal by van Emden \cite{VE86}, which is still appealing because of its neat semantics. Syntactically, our proposal is very close to van Emden's $QLP$: we use qualified definite Horn clauses $A \qgets{d} \ats{B}$ with an attenuation value $d$ attached to the implication and no annotations attached to the atoms. However, we improve \cite{VE86} in the two ways summarized in the abstract: firstly, we replace numeric certainty values (in particular, those playing the role of attenuation factors in program clauses) by  qualification values belonging to a parametrically given {\em Qualification Domain} $\qdomm$ with a lattice structure, which provides abstract operations generalizing the use of $min$ (minimum) and $\times$ (product) in \cite{VE86}. In this way we get a {\em generic scheme} $\qlp{\qdomm}$. Secondly, we present stronger semantic results and a sound and strongly complete goal solving  procedure called \emph{Qualified $SLD$ Resolution} over $\qdomm$(in symbols, $\sld{\qdomm}$), which extends $SLD$ resolution using {\em annotated atoms} and {\em qualification constraints} over $\qdomm$. The $\qlp{\qdomm}$ scheme enjoys nice semantic properties and has interesting instances that can be efficiently implemented using $CLP$ technology: $\qlp{\qdomm}$ programs and goals can be easily translated into $\clp{\cdom{\qdomm}}$ for any choice of a constraint domain $\cdom{\qdomm}$ able to compute with qualification constraints
over $\qdomm$.


We have developed a prototype implementation of some instances of the $\qlp{\qdomm}$ scheme
(including  van Emden's $QLP$) on top of the $CFLP$ system $\toy$.

After this introduction, the rest of the report is structured as follows: Section \ref{Domains} presents the axioms for qualification domains $\qdomm$, showing some basic instances and proving that the class of such domains is closed under cartesian product. Section \ref{Language} presents the syntax and declarative semantics of the $\qlp{\qdomm}$ scheme. Section \ref{Solving} presents  qualified $SLD$ resolution over $\qdomm$ with its soundness and strong completeness properties. Section \ref{Implementation} presents the general implementation technique for $\qlp{\qdomm}$
that we have used to implement some instances  of the scheme (including  van Emden's $QLP$) on top of the $CFLP$ system $\toy$. Finally, Section \ref{Conclusions} presents our conclusions and plans for future work. Appendix \ref{Proofs} includes detailed proofs for the main results. Other proofs that have been ommitted or sketched can be found in \cite{Rom07} (in Spanish).

%% file: TR_SIC-1-08_2-Domains.tex
\section{Qualification Domains} \label{Domains}


By definition, a  {\em Qualification Domain} is  any structure $\qdomm = \langle D, \sqsubseteq, \bot, \top, \circ \rangle$ such that:

\begin{enumerate}
    \item $\langle \qdomm, \sqsubseteq, \bot, \top \rangle$ is a lattice with extreme points $\bot$ and $\top$ w.r.t. the partial ordering $\sqsubseteq$. For given elements  $d, e \in D$, we  write $d\, \sqcap\, e$ for the {\em greatest lower bound} ($glb$) of $d$ and $e$ and $d\, \sqcup\, e$ for the {\em least upper bound} ($lub$) of $d$ and $e$. We also write $d \sqsubset e$ as abbreviation for $d \sqsubseteq e\, \land\, d \neq e$.
    \item $\circ : D \times D \rightarrow D$, called {\em attenuation operation}, verifies the following axioms:
        \begin{enumerate}
            \item $\circ$ is associative, commutative and monotonic w.r.t. $\sqsubseteq$.
            \item $\forall d \in D:\, d \circ \top = d$.
            \item $\forall d \in D:\, d \circ \bot = \bot$.
            \item $\forall d, e \in D \setminus \{\bot,\top\}:\, d \circ e\, \sqsubset\, e$.
            \item $\forall d, e_1, e_2 \in D:\, d  \circ (e_1 \sqcap e_2) = d \circ e_1\, \sqcap\, d \circ e_2$.
        \end{enumerate}
\end{enumerate}


In the rest of the report, $\qdomm$ will generally denote an arbitrary qualification domain. For any finite $S = \{e_{1}, e_{2}, \ldots, e_{n}\} \subseteq D$, the $glb$ of $S$ (noted as $\bigsqcap S$) exists and can be computed as $e_{1} \sqcap e_{2} \sqcap \cdots \sqcap e_{n}$ (which reduces to $\top$ in the case $n = 0$). As an easy consequence of the axioms, one gets the identity $d \circ \bigsqcap S =  \bigsqcap \{d \circ e \mid e \in S\}$. We generalize van Emden's $QLP$ to a generic scheme $\qlp{\qdomm}$ which uses qualification values $d \in D \setminus \{\bot\}$ in place of certainty values $d \in (0,1]$, the $glb$ operator $\bigsqcap$ in place of the minimum operator $min$, and the attenuation operator $\circ$ in place of the multiplication operator $\times$. Three interesting instances of qualification domains are shown below.

\noindent\textbf{The Domain of Classical Boolean Values:} $\B = (\{0,1\}, \le, 0, 1, \land)$, where $0$ and $1$ stand for the two classical truth values  \emph{false} and \emph{true}, $\leq$ is the usual numerical ordering over $\{0,1\}$, and $\land$ stands for the classical conjunction operation over $\{0,1\}$. The instance $\qlp{\B}$ of our $\qlp{\qdomm}$ scheme behaves as classical Logic Programming.

\noindent\textbf{The Domain  of van Emden's Uncertainty Values:} $\U = (\mbox{U}, \leq, 0, 1,\times)$, where $\mbox{U} = [0,1] = \{d \in \REAL \mid 0 \le d \le 1\}$, $\le$ is the usual numerical ordering, and $\times$ is the multiplication operation. In this domain,  the top element $\top$ is $1$ and the greatest lower bound $\bigsqcap S$ of a finite $S \subseteq \mbox{U}$ is the minimum value min(S), which is $1$ if $S = \emptyset$. For this reason, the instance $\qlp{\U}$ of our $\qlp{\qdomm}$ scheme behaves as van Emden's $QLP$.

\noindent\textbf{The Domain of Weight Values:} $\W = (\mbox{P}, \ge, \infty, 0, +)$, where $\mbox{P} = [0,\infty] = \{d \in \REAL \cup \{\infty\} \mid d \ge 0\}$, $\geq$ is the reverse of the usual numerical ordering (with $\infty \ge d$ for any $d \in \mbox{P}$), and $+$ is the addition operation (with $\infty + d = d + \infty = \infty$ for any $d \in \mbox{P}$). In this domain,  the top element $\top$ is $0$ and the greatest lower bound $\bigsqcap S$ of a finite $S \subseteq \mbox{P}$ is the maximum value max(S), which is $0$ if $S = \emptyset$. When working in the instance $\qlp{\W}$ of our $\qlp{\qdomm}$ scheme one propagates to a clause head the qualification value $f + b$, where $f$ is the clause's 'attenuation factor' and $b$ is the maximum of the qualification values known for the body atoms. Therefore, qualification values in  the instance $\qlp{\W}$ of our $\qlp{\qdomm}$ scheme behave as a weighted measure of the depth of proof trees.

It is easily checked that the axioms of qualification domains are satisfied by $\B$, $\U$ and $\W$. In fact, the axioms have been chosen as a natural  generalization of some basic properties satisfied by the ordering $\leq$ and the operation $\times$ in $\U$. In general, the values belonging to a qualification domain are intended to qualify logical assertions by measuring the degree in which they satisfy some kind of user's expectations. In this way, one can think of $\U$ values as measuring the degree of truth, $\W$ values as measuring proofs sizes, etc.


Given two qualification domains  $\qdomm_i = \langle D_i, \sqsubseteq_i, \bot_i, \top_i, \circ_i \rangle$ ($i \in \{1, 2\}$), their {\em cartesian product} $\qdomm_1 \times \qdomm_2$ is defined as $\qdomm =_{\mathrm{def}} \langle D, \sqsubseteq, \bot, \top, \circ \rangle$, where  $D =_{\mathrm{def}} D_1 \times D_2$, the partial ordering $\sqsubseteq$ is defined as $(d_1,d_2) \sqsubseteq (e_1,e_2) \diff d_1 \sqsubseteq_1 e_1$ and $d_2 \sqsubseteq_2 e_2$, $\bot =_{\mathrm{def}} (\bot_1, \bot_2)$, $\top =_{\mathrm{def}} (\top_1, \top_2)$, and the attenuation operator $\circ$ is defined as $(d_1,d_2) \circ (e_1,e_2) =_{\mathrm{def}} (d_1 \circ_1 e_1, d_2 \circ_2 e_2)$. The class of the qualification domains is closed under cartesian products, as stated in the following result.

\begin{proposition}\label{propQDomProduct}
The cartesian product $\qdomm = \qdomm_1 \times \qdomm_2$ of two given qualification domains is always another qualification domain.
\end{proposition}

\begin{proof}
According to the axiomatic definition of qualification domains, one must prove two items:

\begin{enumerate}

    \item $\qdomm$ is a lattice with extreme points $\bot$ and $\top$ w.r.t. the partial ordering $\sqsubseteq$. This is easily checked using the definition of $\sqsubseteq$ in the product domain. In particular, one gets the equalities $(d_1,d_2) \sqcap (e_1,e_2) = (d_1 \sqcap_1 e_1,d_2 \sqcap_2 e_2)$ and $(d_1,d_2) \sqcup (e_1,e_2) = (d_1 \sqcup_1 e_1,d_2 \sqcup_2 e_2)$.

    \item $\circ$ satisfies the five axioms required for attenuation operators, i.e.:
        \begin{enumerate}
            \item $\circ$ is associative, commutative and monotonic w.r.t. $\sqsubseteq$.
            \item $\forall (d_1,d_2) \in D_1 \times D_2$ : $(d_1,d_2) \circ \top = (d_1,d_2)$.
            \item $\forall (d_1,d_2) \in D_1 \times D_2$ : $(d_1,d_2) \circ \bot = \bot$.
            \item $\forall (d_1,d_2), (e_1,e_2) \in D_1 \times D_2 \setminus \{\bot,\top\}$ : $(d_1,d_2) \circ (e_1, e_2) \sqsubset (e_1,e_2)$.
            \item $\forall (d_1,d_2), (e_1,e_2), (e'_1, e'_2) \in D_1 \times D_2$ : $(d_1,d_2) \circ ((e_1, e_2) \sqcap (e'_1, e'_2)) = ((d_1, d_2) \circ (e_1, e_2)) \sqcap ((d_1, d_2) \circ (e'_1, e'_2))$.
        \end{enumerate}

All these conditions are easily proved, using the hypothesis that both $\qdomm_1$ and $\qdomm_2$ are qualification domains as well as the construction of $\qdomm$ as cartesian product of $\qdomm_1$ and $\qdomm_2$. \qed
\end{enumerate}
\end{proof}

Intuitively, each value $(d_1,d_2)$ belonging to a product domain $\qdomm_1 \times \qdomm_2$ imposes the qualification $d_1$ {\em and also} the qualification $d_2$. In particular, values $(c,d)$ belonging to the product domain $\U \times \W$ impose two qualifications, namely: a certainty value greater or equal than $c$ and a proof tree with depth less or equal than $d$. These intuitions indeed correspond to the declarative and operational semantics formally defined in Sections \ref{Language} and \ref{Solving}.

%% file: TR_SIC-1-08_3-Language.tex
\section{Syntax and Semantics of QLP(\qdom)} \label{Language}


\subsection{Programs, Interpretations and Models}

We assume a {\em signature} $\Sigma$ providing free function symbols (a.k.a. constructors) and predicate symbols. {\em Terms} are built from constructors and variables from a countably infinite set $\Var$, disjoint from $\Sigma$. {\em Atoms} are of the form $p(t_1,\, \ldots,\, t_n)$ (abbreviated as $p(\overline{t_n})$) where $p$ is a $n$-ary predicate symbol and $t_i$ are terms. We write $At_{\Sigma}$ for the set of all the atoms, called the {\em open Herbrand base}. A $\qlp{\qdomm}$ program $\Prog$ is a set of {\em qualified definite Horn clauses} of the form $A \qgets{d} \ats{B}$ where $A$ is an atom, $\ats{B}$ a finite conjunction of atoms and $d \in D \setminus \{\bot\}$  is the {\em attenuation value} attached to the clause's implication. In $\qlp{\B}$ programs, the only choice for $d$ is $1$, standing for $true$, and therefore $\qlp{\B}$ behaves as classical $LP$. The following example presents two simple programs over the domains $\U$ and $\W$. It is not meant as a realistic application, but just as an illustration.

\begin{example}\label{exmpProg}
\hfill
\begin{enumerate}
\item The $\qlp{\U}$ program $\UProg$ displayed below can be understood as a {\em knowledge base} given by the facts for the predicates \texttt{animal}, \texttt{plant}, \texttt{human} and \texttt{eats}, along with {\em knowledge inference rules} corresponding to the clauses with non-empty body. The clauses for the predicate \texttt{human} specify the  human beings as the ancestors of \texttt{adam} and \texttt{eve}, with the certainty of being an actual human decreasing as one moves back along the ancestors' chain. Therefore, the certainty of being a cruel human also decreases when moving from descendants to ancestors.

\begin{verbatim}
  cruel(X) <-0.90- human(X), eats(X,Y), animal(Y)
  cruel(X) <-0.40- human(X), eats(X,Y), plant(Y)

  animal(bird) <-1.0-       human(adam) <-1.0-
  animal(cat) <-1.0-        human(eve) <-1.0-
  plant(oak) <-1.0-         human(father(X)) <-0.90- human(X)
  plant(apple) <-1.0-       human(mother(X)) <-0.90- human(X)

  eats(adam, X) <-0.80-
  eats(eve,X) <-0.30- animal(X)
  eats(eve,X) <-0.60- plant(X)
  eats(father(X),Y) <-0.80- eats(X,Y)
  eats(mother(X),Y) <-0.70- eats(X,Y)
\end{verbatim}

\item The $\qlp{\W}$ program $\WProg$ is very similar to $\UProg$, except that the attenuation value \texttt{1} is attached to all the clauses. Therefore, each clause is intended to convey the information that the depth of a proof tree for the head is 1 plus the maximum depth  of proof trees for the atoms in the body. As we will see, qualification constraints over $\W$ can be used to impose upper bounds to the depths of proof trees when solving goals w.r.t. $\WProg$.
\end{enumerate}

\noindent Note that the two programs in this example are different qualified versions of the classical $LP$ program $\Prog$ obtained by dropping all the annotations. Due to the left recursion in the clauses for the predicates  \texttt{human} and \texttt{eats}, some goals for $\Prog$ have an infinite search space where $SLD$ resolution with a leftmost selection strategy would fail to compute some expected answers. For instance, the answer $\{\texttt{X} \mapsto \texttt{mother(eve)}, \texttt{Y} \mapsto \texttt{apple}\}$ would not be computed for the goal \texttt{eats(X,Y)}. However, when solving goals for the qualified programs  $\UProg$ and $\WProg$  using the resolution method presented in Section \ref{Solving}, qualification constraints can be used for imposing bounds to the search space, so that even the leftmost selection strategy leads to successful computations. \qed
\end{example}

As shown in the example, clauses contain classic atoms in both their head and their body. But for our semantics, we will be interested in not only proving that we can infer an atom for a given program, but proving that we can infer it with qualification greater or equal than some given  value. For this reason, we introduce {\em \qdom-annotated atoms} $\at{A}{d}$, consisting of  an atom $A$ with an attached `annotation' $d \in \aqdomm$. For use in goals to be solved, we consider also {\em open annotated atoms}  of the form $\at{A}{W}$, where $W$ is a {\em qualification variable} intended to take values over $\aqdomm$. We postulate a countably infinite set $\War$ of qualification variables, disjoint from $\Var$ and $\Sigma$.

The {\em annotated Herbrand base} over $\qdomm$ is defined as the set $At_{\Sigma}(\qdomm)$ of all $\qdomm$-annotated atoms. The {\em \qdom-entailment relation} over  $At_{\Sigma}(\qdomm)$ is defined as follows: $\at{A}{d} \Dentail \at{A'}{d'}$ iff  there is some substitution $\theta$ such that $A' = A\theta$ and $d' \sqsubseteq d$. Finally, we define an {\em open Herbrand interpretation} over $\qdomm$ as any subset $\I \subseteq At_{\Sigma}(\qdomm)$ which is closed under \qdom-entailment. That is, an open Herbrand interpretation $\I$ including a given annotated atom $\at{A}{d}$ is required to include all the `instances' $\at{A'}{d'}$ such that $\at{A}{d} \Dentail \at{A'}{d'}$, because we intend to formalize a semantics such that all such instances are valid whenever $\at{A}{d}$ is valid.

In the sequel we refer to open Herbrand interpretations just as Herbrand interpretations, and we write $\intd$ for the family of all Herbrand interpretations over $\qdomm$. The following proposition is easy to prove from the definition of a Herbrand interpretation and the definitions of the union and intersection of a family of sets.

\begin{proposition}\label{propQDomLattice}
The family $\intd$ of all Herbrand interpretations over $\qdomm$ is a complete lattice under the inclusion ordering $\subseteq$, whose extreme points are $\intd$ as maximum and $\emptyset$ as minimum. Moreover, given any family of interpretations $I \subseteq \intd$, its $lub$ and $glb$ are $\supr I = \union \{\I \in \intd \mid \I \in I\}$ and $\infi I = \inter \{\I \in \intd \mid \I \in I\}$, respectively. \qed
\end{proposition}

Let $C$ be any clause $A \qgets{d} B_{1}, \ldots, B_{k}$ in the program $\Prog$, and $\I \in \intd$ any interpretation over $\qdomm$. We say that $\I$ is a {\em model} of $C$ if and only if for any substitution $\theta$ and any qualification values $d_{1}, \ldots, d_{k} \in D \setminus \{\bot\}$ such that $\at{B_{i}\theta}{d_{i}} \in \I$ for all $1 \le i \le k$, one has $\at{A\theta}{(d \circ \infi\{d_{1}, \ldots, d_{k}\})} \in \I$. And we say that $\I$ is a model of the $\qlp{\qdomm}$ program $\Prog$ (in symbols, $\I \models \Prog$) if and only if $\I$ is a model of each clause in $\Prog$.

\subsection{Declarative Semantics}

As in any logic language, we need some technique to infer formulas (in our case, \qdom-annotated atoms) from  a given $\qlp{\qdomm}$ program $\Prog$. Following traditional ideas, we consider two alternative ways of formalizing an inference step which goes from the body of a clause to its head: an operator $\Tp$ and a qualified variant of Horn Logic, noted as $\qhl{\qdomm}$ and called \emph{Qualified Horn Logic}. The operator $\Tp : \intd \to \intd$ is defined as:
\[\begin{array}{ll}
\Tp(\I) =_{\mathrm{def}} \{ \at{A'}{d'} \mid & (A \qgets{d} B_{1}, \ldots, B_{k})\in\Prog,\\
& \theta\mbox{ subst.},\, \at{B_{i}\theta}{d_{i}} \in \I \mbox{ for all } 1 \le i \le k,\, A' = A\theta,\\
& d' \in D \setminus \{\bot\},\, ºd' \sqsubseteq d \circ \infi \{d_{1}, \ldots, d_{k}\}\}
\end{array}\]
Intuitively, we can see that for a given interpretation $\I$, $\Tp(\I)$ is the set of those \qdom-annotated atoms obtained by considering $\qdomm$-annotated bodies of clause instances that are included in $\I$ and propagating an annotation to the head via the clause's qualification value.

The logic $\qhl{\qdomm}$ is defined as a deductive system consisting just of one  inference rule $\mbox{QMP}(\qdomm)$, called \emph{Qualitative Modus Ponens} over $\qdomm$. If there are some $(A \qgets{d} B_{1}, \ldots, B_{k}) \in \Prog$, some substitution $\theta$ such that $A' = A\theta$ and $B'_{i} = B_{i}\theta$ for all $1 \le i \le k$ and $d' \sqsubseteq d \circ \infi\{d_{1}, \ldots, d_{k}\}$, the following inference step is allowed:

\[
    \frac
    {\quad \at{B'_{1}}{d_{1}} \quad \cdots \quad \at{B'_{k}}{d_{k}} \quad}
    {\at{A'}{d'}} \quad \mbox{QMP}(\qdomm)
\]
We will use the notations $\Prog \qhld \at{A}{d}$ (resp. $\Prog \qhldn{n} \at{A}{d}$) to indicate that $\at{A}{d}$ can be inferred from the clauses in program $\Prog$ in finitely many steps (resp. $n$ steps). Note that  $\qhl{\qdomm}$ proofs can be naturally represented as upwards growing {\em proof trees}
with \qdom-annotated atoms at their nodes, each node corresponding to one inference step having the children nodes as premises.

The following proposition collects the main results concerning the declarative semantics of the $\qlp{\qdomm}$ scheme. We just sketch some  key proof ideas. The  full proofs are given in \cite{Rom07}. As in \cite{VE86}, they can be developed in analogy  to the classical papers  \cite{VEK76,AVE82}, except that our Herbrand interpretations are open, as first suggested by Clark in \cite{Cla79}. Our use of the $\qhl{\qdomm}$ calculus is obviously related to the classical $\Tp$ operator, although it has no direct counterpart in the historical papers we are aware of.

\begin{proposition}\label{propResuls}
The following assertions hold for any $\qlp{\qdomm}$ program $\Prog$:
\begin{enumerate}
    \item $\I \models \Prog \iff \Tp(\I) \subseteq \I$ .
    \item $\Tp$ is monotonic and continuous.
    \item The least fixpoint  $\mu(\Tp)$ is the least Herbrand model of $\Prog$, noted as $\Mp$.
    \item $\Mp = \union_{n\in\NAT} \Tp\uparrow^{n}(\emptyset) = \{\at{A}{d} \mid \Prog \qhld \at{A}{d}\}$.
\end{enumerate}
\end{proposition}

\begin{proof}[Sketch]
Item (1) is easy to prove from the definition of $\Tp$. In item (2), monotonicity ($\I \subseteq \J \Longrightarrow \Tp(\I) \subseteq \Tp(\J)$) follows easily from the definition of $\Tp$ and continuity ($\Tp(\union_{n\in\NAT} \I_{n}) = \union_{n\in\NAT} \Tp(\I_{n})$ for any chain $\{\I_{n} \mid n\in\NAT\} \subseteq \intd$ with $\I_{n} \subseteq \I_{n+1}$ for all $n\in\NAT$) follows from monotonicity and properties of chains and sets of interpretations. Item (3) follows from (1), (2), Proposition \ref{propQDomLattice} and some known properties about lattices. Finally, item (4) follows from proving the two implications $\Prog \qhldn{n} \at{A}{d} \Longrightarrow \exists m\, (\at{A}{d} \in \Tp\uparrow^{m}(\emptyset))$ and
 $\at{A}{d} \in \Tp\uparrow^{n}(\emptyset) \Longrightarrow \exists m\, (\Prog \qhldn{m} \at{A}{d})$ by induction on $n$. \qed
\end{proof}

The next example presents  proofs deriving annotated atoms that belong to the least models of the programs $\UProg$ and $\WProg$ from Example \ref{exmpProg}.

\begin{example}\label{exmpMp}
\hfill
\begin{enumerate}
\item The proof tree displayed below shows that the $\U$-annotated atom  at its root can be deduced from $\UProg$ in $\qhl{\U}$. Therefore, the atom belongs to $\M{\UProg}$.
    {\scriptsize{\[\frac
        {\displaystyle\frac
            {\displaystyle\frac{}{\,\texttt{human(eve)\#1.0}\,}}
            {\, \texttt{human(mother(eve))\#0.90} \,} \quad
        \displaystyle\frac
            {\displaystyle\frac
                {\displaystyle\frac{}{\,\texttt{animal(bird)\#1.0}\,}}
                {\,\texttt{eats(eve,bird)\#0.30}\,}}
            {\, \texttt{eats(mother(eve),bird)\#0.21} \,} \quad
        \displaystyle\frac
            {}
            {\, \texttt{animal(bird)\#1.0} \,}}
        {\quad \texttt{cruel(mother(eve))\#0.15} \quad}\]}}

It is easy to see which clause was used in each inference step. Note that the atom at the root could have been proved even with the greater certainty value  \texttt{0.189}. However, since \texttt{0.15} $\le$ \texttt{0.189},  the displayed inference it is also correct (albeit less informative).

\item A proof tree quite similar to the previous one, but  with different annotations, can be easily built to show that \texttt{cruel(mother(eve))\#4} can be deduced from $\WProg$ in $\qhl{\W}$. Therefore, this annotated atom belongs to $\M{\WProg}$. It conveys the information that \texttt{cruel(mother(eve))} has a proof tree of depth  \texttt{4} w.r.t. to the classical $LP$ program $\Prog$ obtained by dropping $\WProg$'s annotations. \qed
\end{enumerate}
\end{example} 

%% file: TR_SIC-1-08_4-Solving.tex
\section{Goal Solving by SLD(\qdom) Resolution} \label{Solving}

\subsection{Goals and Solutions}

In classical logic programming a goal is presented as a conjunction of atoms. In our setting, proving atoms with arbitrary qualifications may be unsatisfactory, since qualification values too close to $\bot$ may not ensure sufficient information. For this reason, we present goals as conjunctions of open \qdom-annotated atoms and we indicate the minimum qualification value required  each of them. Hence initial goals look like: $\at{A_{1}}{W_{1}}, \ldots, \at{A_{n}}{W_{n}} \sep W_{1} \sqsupseteq \beta_{1}, \ldots, W_{n} \sqsupseteq \beta_{n}$, where $W_i \in \War$ and $\beta_i \in D \setminus \{\bot\}$. Observe that we have annotated all atoms in the goal with qualification variables $W_i$  instead of plain values because we are interested in any solution that satisfies the {\em qualification constraints}  $W_{i} \sqsupseteq \beta_{i}$, used to impose lower bounds to the atoms' qualifications.

As explained in the next Subsection, goal resolution proceeds from an initial goal through intermediate goals until reaching a final solved goal. The intermediate goals have a more general form, consisting of a composition of three items: a conjunction of \qdom-annotated atoms $\ats{A}$ waiting to be solved, a substitution $\sigma$ computed in previous steps, and a set of qualification constraints $\Delta$. We consider two kinds of qualification constraints:

\begin{enumerate}
    \item $\alpha \circ W \sqsupseteq \beta$,
     where $W \in \War$ is qualification variable and $\alpha, \beta \in D \setminus \{\bot\}$ are such that $\alpha \sqsupseteq \beta$.
    This  is called a \emph{threshold constraint} for $W$.
    \item $W = d \circ \infi\{W_{1}, \ldots, W_{k}\}$,
    where $W, W_{1}, \ldots, W_{k} \in \War$ are qualification variables and $d \in D \setminus \{\bot\}$.
    This  is called a \emph{defining constraint} for $W$.
\end{enumerate}

In order to understand why these two kinds of constraints are needed, think of an annotated atom  $\at{A}{W}$ within an initial goal which includes also an initial threshold constraint $\top \circ W \sqsupseteq \beta$ (i.e. $W \sqsupseteq \beta$) for $W$. Applying a resolution step with  a program clause whose head unifies with $A$ and whose attenuation value is $d \in D \setminus \{\bot\}$ will lead to a new goal including a defining constraint $W = d \circ \infi\{W_{1}, \ldots, W_{k}\}$ for $W$ and a threshold constraint $d\circ\top\circ W_{i} \sqsupseteq \beta$ for each $1 \le i \le k$,
where the new qualification variables $W_i$  correspond to the atoms in the clause's body. This explains the need to introduce defining constraints as well as more general threshold constraints $\alpha \circ W \sqsupseteq \beta$. Intuitively, the  values $\alpha$ and $\beta$ within such constraints
play the role of an {\em upper} and a {\em lower} bound, respectively. As we will see, our goal solving procedure takes advantage of these bounds for pruning useless parts of the computation search space.

Let us now present some notations needed  for a formal definition of goals. Given a conjunction of \qdom-annotated atoms $\ats{A}$ and a set of qualification constraints $\Delta$, we define the following sets of variables:
\begin{itemize}
    \item $\mbox{var}(\ats{A}) =_{\mathrm{def}} \union \{\mbox{var}(A) \mid \at{A}{W} \in \ats{A}\}$ .
    \item $\mbox{war}(\ats{A}) =_{\mathrm{def}} \union \{W \mid \at{A}{W} \in \ats{A}\}$ .
    \item $\mbox{war}(\Delta)$ as the set of qualification variables that appears in any qualification constraint in $\Delta$.
    \item $\mbox{dom}(\Delta)$ as the set of qualification variables that appear in the left hand side
    of any qualification constraint in $\Delta$.
\end{itemize}
We also say that $\Delta$ is \emph{satisfiable} iff there is some $\omega \in \sustd$ --the set of all the substitutions of values in $D \setminus \{\bot\}$ for variables in  $\War$-- such that $\omega \in Sol(\Delta)$, what means that $\omega$ satisfies every qualification constraint in $\Delta$, i.e. $\omega$ is a {\em solution} of $\Delta$. Moreover, we say that $\Delta$ is \emph{admissible} iff it satisfies the following three conditions:

\begin{enumerate}
    \item $\Delta$ is satisfiable,
    \item for every $W \in \mbox{war}(\Delta)$ there exists one and only one constraint for $W$ in $\Delta$ (this implies $\mathrm{dom}(\Delta) = \mathrm{war}(\Delta)$), and
    \item the relation $>_{\Delta}$ defined by $W >_{\Delta} W_{i}$ iff
    there is some defining constraint $W = \alpha \circ \infi\{W_{1}, \ldots, W_{i}, \ldots, W_{k}\}$ in $\Delta$,
    satisfies that $>^{*}_{\Delta}$ is irreflexive.
\end{enumerate}
Finally, we say that $\Delta$  is \emph{solved} iff $\Delta$ is admissible and only contains defining constraints.
Now we are in a position to define {\em goals} and their {\em solutions}:

\begin{definition}[Goals and its Variables] \label{defGoal}
Given a conjunction of \qdom-annotated atoms $\ats{A}$, a substitution $\sigma \in \sust$ --the set of all substitutions of terms for variables in $\Var$-- and a set of qualification constraints $\Delta$, we say that $G \equiv \ats{A} \sep \sigma \sep \Delta$ is a \emph{goal} iff
    \begin{itemize}
        \item[i.] $\sigma \in \sust$ is idempotent and such that $dom(\sigma) \cap var(\ats{A}) = \emptyset$.
        \item[ii.] $\Delta$ is admisible.
        \item[iii.] For every qualification variable in $\mathrm{war}(\ats{A})$ there is one and only one threshold constraint for $W$ in $\Delta$. And there are no more threshold constraints in $\Delta$.
    \end{itemize}
Furthermore, if $\sigma = \epsilon$ (the identity substitution) then $G$ is called \emph{initial}, and if $\ats{A}$ is empty and $\Delta$ is solved, then $G$ is called \emph{solved}. For any goal $G$, we define the set of variables of $G$ as $\mathrm{var}(G) =_{\mathrm{def}} \mathrm{var}(\ats{A}) \cup \mathrm{dom}(\sigma)$ and the set of qualification variables of $G$ as $\mathrm{war}(G) =_{\mathrm{def}} \mathrm{war}(\ats{A}) \cup \mathrm{dom}(\Delta)$.
\qed
\end{definition}

\begin{definition}[Goal Solutions]\label{defSolutions}
A pair of substitutions $(\theta, \rho)$ such that $\theta \in \sust$ and $\rho \in \sustd$ is called a \emph{solution} of a goal $G \equiv \ats{A} \sep \sigma \sep \Delta$ iff:
\begin{enumerate}
    \item $\theta = \sigma\theta$ .
    \item $\rho \in Sol(\Delta)$ .
    \item $\Prog \qhld \at{A\theta}{W\rho}$ for all $\at{A}{W} \in \ats{A}$ .
\end{enumerate}

In addition, a solution $(\sigma,\mu)$ for a goal $G$ is said to be more general than another solution $(\theta,\rho)$ for the same goal $G$ (one also says in this case  that $(\theta,\rho)$ is subsumed by $(\sigma,\mu)$) iff $\sigma \preccurlyeq \theta$ $[\varset{G}]$ and $\mu \sqsupseteq \rho$ $[\warset{G}]$, where $\sigma \preccurlyeq \theta$ $[\varset{G}]$ means that there is some substitution $\eta$ such that the composition $\sigma\eta$ behaves the same as $\theta$ over any variable in the set $\varset{G}$ and $\mu \sqsupseteq \rho$ $[\warset{G}]$ means that $\mu(W) \sqsupseteq \rho(W)$ holds for any $W \in \warset{G}$.
\qed
\end{definition}

Any solved goal $G' \equiv \sigma \sep \Delta$ has the  \emph{associated solution} $(\sigma, \mu)$, where  $\mu = \omega_{\Delta}$ is the qualification substitution given by $\Delta$, such that $\omega_{\Delta}(W)$ is the qualification value determined by the defining constraints in $\Delta$ for all $W \in \mbox{dom}(\Delta)$, and $\omega_{\Delta}(W) = \bot$ for any $W \in \War \setminus \mbox{dom}(\Delta)$. Note that for any $W \in \mbox{dom}(\Delta)$ there exists one unique defining constraint  $W = d \circ \infi\{W_{1}, \ldots, W_{k}\}$ for $W$ in $\Delta$, and then  $\omega_{\Delta}(W)$ can be recursively computed as  $d \circ \infi\{\omega_{\Delta}(W_{1}), \ldots, \omega_{\Delta}(W_{k})\}$. The solutions associated to solved goals are called {\em computed answers}.

\begin{example}\label{exmpGoalSol}
\hfill
\begin{enumerate}
    \item A possible goal for program $\UProg$ in Example \ref{exmpProg} is \texttt{eats(father(X),Y)\#W1, human(father(X))\#W2 | W1>=0.4, W2>=0.6}; and a valid solution for it is \texttt{\{X $\mapsto$ adam, Y $\mapsto$ apple\} | \{W1 $\mapsto$ 0.50, W2 $\mapsto$ 0.75\}}.
    \item A goal for program $\WProg$ in Example \ref{exmpProg} may be \texttt{eats(X,Y)\#W | W<=5.0}; and a valid solution is \texttt{\{X $\mapsto$ father(adam), Y $\mapsto$ apple\} | \{W $\mapsto$ 4.0\}}.
    \qed
\end{enumerate}
\end{example}

Note that the goal for $\UProg$ in the previous example imposes lower bounds to the certainties to be computed, while the goal for $\WProg$ imposes an upper bound to the proof depth. In general, goal solving in $\qlp{\W}$ corresponds to depth-bound goal-golving in classical Logic Programming.

\subsection{SLD(\qdom) Resolution}

We propose a sound and strongly complete goal solving procedure called \emph{Qualified $SLD$ Resolution} parameterized over a given qualification domain $\qdomm$, written as $\sld{\qdomm}$, which makes use of \emph{annotated atoms} and \emph{qualification constraints} over $\qdomm$. The implementation of this goal solving procedure using $CLP$ technology will be discussed in the next section. Resolution computations are written $G_{0} \qres{C_{1}, \sigma_{1}} G_{1} \qres{C_{2}, \sigma_{2}} \cdots \qres{C_{n}, \sigma_{n}} G_{n}$, abbreviated as $G_{0} \qresn{*}{\sigma} G_{n}$ with $\sigma = \sigma_{1}\sigma_{2}\cdots\sigma_{n}$. They are finite sequences of resolution steps $G_{i-1} \qres{C_{i}, \sigma_{i}} G_{i}$, starting with an initial goal $G_0$ and ending up with a solved goal $G_n$. One single resolution step is formally defined as follows:

\begin{definition}[Resolution step] \label{defStep}
A \emph{resolution step} has the form \\
$\ats{L},$ $\at{A}{W}, \ats{R} \sep \sigma \sep \alpha \circ W \sqsupseteq \beta, \Delta \qres{C_{1}, \sigma_{1}} (\ats{L},\at{B_{1}}{W_{1}}, \ldots \at{B_{k}}{W_{k}}, \ats{R})\sigma_{1} \sep \sigma\sigma_{1} \sep \Delta_{1}$ \\
where $\at{A}{W}$ is  the {\em selected atom}, $\Delta_{1} = d\circ\alpha\circ W_{1} \sqsupseteq \beta, \ldots, d\circ\alpha\circ W_{k} \sqsupseteq \beta, W = d\circ \infi\{W_{1}, \ldots, W_{k}\}, \Delta$, $C_{1} \equiv (H \qgets{d} B_{1}, \ldots, B_{k}) \in_{\mathrm{var}} \Prog$  is chosen as a variant of a clause in $\Prog$ with fresh variables and such that $d \circ \alpha \sqsupseteq \beta$, $\sigma_{1}$ is the m.g.u. between $A$ and $H$, and $W_{1}, \ldots, W_{k} \in \War$ are fresh qualification variables. \qed
\end{definition}

The notation $\alpha \circ W \sqsupseteq \beta, \Delta$ represents a set of qualification constraints including the threshold constraint $\alpha \circ W \sqsupseteq \beta$ plus those in $\Delta$, with no particular ordering assumed. Notice that the condition $d \circ \alpha \sqsupseteq \beta$ is required for the resolution step to be enabled. In this way, threshold constraints $\alpha \circ W \sqsupseteq \beta$ are actively used for pruning parts of the computation search space where no solutions can be found. In the instance of $\qlp{\B}$ it is easily checked that all the qualification values and constraints become trivial, so that $\sld{\B}$ boils down to classical $SLD$ resolution. In the rest of this section we present the main properties of $\sld{\qdomm}$ resolution in the general case.

\begin{proposition}\label{propResolution}
If $G$ is a goal and $G_{0} \qres{C_{1},\sigma_{1}} G_{1}$, then $G_{1}$ is also a goal.
\end{proposition}

\begin{proof}[Sketch]
Assume a goal $G_{0}$ and a $\sld{\qdomm}$ resolution step $G_{0} \qres{C_{1},\sigma_{1}} G_{1}$, as in Definition \ref{defStep}. Then $G_0$ satisfies the conditions required for goals in Definition \ref{defGoal}, and we must show that $G_{1}$ also satisfies such conditions. This is not difficult to check, using the fact that $C_{1}$ has been chosen without variables in common with $G_{0}$. In particular, note that the threshold constraint for $W$ in $G_{0}$ is absent in $G_{1}$, which includes a defining constraint for $W$ and threshold constraints for the new qualification variables $W_{i}$. \qed
\end{proof}

The next two theorems are the main theoretical results in this report. The Soundness Theorem \ref{thmSoundness} guarantees  that every computed answer is correct in the sense that it is a solution of the given goal. The Strong Completeness Theorem \ref{thmCompleteness} ensures that, for any solution of a given goal and any fixed selection strategy, $\sld{\qdomm}$ resolution is able to compute  an equal, if not better, solution.
The proofs, given in Appendix \ref{Proofs}, use inductive techniques similar to those presented in \cite{Sta90} for classical $SLD$ resolution.
Example \ref{exmpRes} below illustrates the Completeness Theorem.

\begin{theorem}[Soundness]\label{thmSoundness}
Assume $G_{0} \qresn{*}{} G$ and $G = \sigma \sep \Delta$ solved. Let $(\sigma,\mu)$ be the solution associated to $G$. Then $(\sigma, \mu)$ --called the computed answer-- is a solution of $G_{0}$. \qed
\end{theorem}

\begin{theorem}[Strong Completeness]\label{thmCompleteness}
Assume a given solution $(\theta,\rho)$ for $G_{0}$ and any fixed strategy for choosing the selected atom at each resolution step. Then there is some computed answer $(\sigma, \mu)$ for $G_{0}$ which subsumes $(\theta,\rho)$. \qed
\end{theorem}

\begin{example}\label{exmpRes}
\hfill
\begin{enumerate}
    \item The following $\sld{\U}$ computation solves the goal for program $\UProg$ presented in Example \ref{exmpGoalSol}:
    \begin{center}
        \begin{tabular}{ll}
            \texttt{eats(father(X),Y)\#W1,} & \\
            $\qquad$ \texttt{human(father(X))\#W2 |} & \\
            $\qquad$ \texttt{W1 >= 0.4, W2 >= 0.6} & $\qres{eats.4, \{X \mapsto \texttt{adam}\}}$ \\
            \texttt{eats(adam,Y)\#W3,} & \\
            $\qquad$ \texttt{human(father(adam))\#W2 |} $\{X \mapsto \texttt{adam}\}$ \texttt{|} & \\
            $\qquad$ \texttt{W1 = 0.8 * min\{W3\},} & \\
            $\qquad$ \texttt{W2 >= 0.6, 0.8 * W3 >= 0.4 } & $\qres{eats.1, \epsilon}$ \\
            \texttt{human(father(adam))\#W2 |} $\{X \mapsto \texttt{adam}\}$ \texttt{|} & \\
            $\qquad$ \texttt{W1 = 0.8 * min\{W3\},} & \\
            $\qquad$ \texttt{W2 >= 0.6, W3 = 0.8} & $\qres{human.3, \epsilon}$ \\
            \texttt{human(adam)\#W4 |} $\{X \mapsto \texttt{adam}\}$ \texttt{|} & \\
            $\qquad$ \texttt{W1 = 0.8 * min\{W3\},} & \\
            $\qquad$ \texttt{W2 = 0.9 * min\{W4\},} &  \\
            $\qquad$ \texttt{W3 = 0.8, 0.90 * W4 >= 0.6 } & $\qres{human.1, \epsilon}$ \\
            \texttt{|} $\{X \mapsto \texttt{adam}\}$ \texttt{|} & \\
            $\qquad$ \texttt{W1 = 0.8 * min\{W3\},} & \\
            $\qquad$ \texttt{W2 = 0.9 * min\{W4\},} &  \\
            $\qquad$ \texttt{W3 = 0.8, W4 = 1.0} &
        \end{tabular}
    \end{center}

    Note that the computed answer \texttt{\{X $\mapsto$ adam\} | \{W1 $\mapsto$ 0.64, W2 $\mapsto$ 0.90\}} subsumes the solution for the same goal given in Example \ref{exmpGoalSol}.

    \item Similarly, $\sld{\W}$ resolution can solve the goal \texttt{eats(X,Y)\#W | W <= 5.0} for $\WProg$, obtaining a computed answer  \texttt{\{X $\mapsto$ father(adam)\} | \{W $\mapsto$ 3.0\}} which subsumes the solution for the same goal given in Example \ref{exmpGoalSol}. \qed
\end{enumerate}
\end{example}

%% file: TR_SIC-1-08_5-Implementation.tex
\section{Towards an Implementation} \label{Implementation}

In this section we assume a qualification domain $\qdomm$ and a constraint domain $\cdom{\qdomm}$ such that the qualification constraints used in $\sld{\qdomm}$ resolution can be expressed as  $\cdom{\qdomm}$ constraints, and we describe a translation of $\qlp{\qdomm}$ programs $\Prog$ and goals $G$ into $\clp{\cdom{\qdomm}}$ programs $\Prog^{t}$ and goals $G^{t}$, such that solving $G$ with  $\sld{\qdomm}$ resolution using $\Prog$ corresponds to solving  $G^{t}$ with constrained $SLD$ resolution using $\Prog^{t}$ and a solver for $\cdom{\qdomm}$.

The translation can be used to develop an implementation of $\sld{\qdomm}$ resolution for the $\qlp{\qdomm}$ language on top of any $CLP$ or $CFLP$ system that supports $\cdom{\qdomm}$ constraints. In particular, if $\qdomm$ is any of the two qualification domains $\U$ or $\W$, the constraint domain $\cdom{\qdomm}$ can be chosen as $\rdom$, which supports arithmetic constraints over the real numbers \cite{JMSY92}. We have developed prototype implementations for $\qlp{\U}$, $\qlp{\W}$ and $\qlp{\U \times \W}$ on top of the $CFLP$ system $\toy$ \cite{toy}, that supports $\rdom$ constraints. Note that although the use of a $CLP(\rdom)$ system could lead to a more efficient implementation, we have chosen a $CFLP(\rdom)$ system instead of a $CLP(\rdom)$ one due to our interest in a future extension of
the $\qlp{\qdomm}$ scheme  to support qualified  $CFLP$ programming. 

Our translation of a  $\qlp{\qdomm}$ program works by adding three extra arguments to all predicates and translating each clause independently. Given the $\qlp{\qdomm}$ clause \[C \equiv p(\ats{t}) \qgets{d} q_{1}(\ats{s_{1}}), \ldots, q_{k}(\ats{s_{k}})\] its head is translated as $p(\ats{t}, Alpha, W, Beta)$, where the new variables $Alpha$, $W$ and $Beta$ correspond, respectively, to $\alpha$, $W$ and $\beta$ in the  threshold constraint $\alpha \circ W \sqsupseteq \beta$ related to a \qdom-annotated atom $\at{A}{W}$ which could be selected for a $\sld{\qdomm}$ resolution step using the clause $C$. The clause's body is translated with the aim of  emulating such a resolution step, and the translated clause becomes:

\begin{center}
\begin{tabular}{lcl}
$C^{t} \equiv p(\ats{t}, Alpha, W, Beta)$ & $\leftarrow$ & $d\circ Alpha \sqsupseteq Beta,$ \\
    & & $W_{1} \sqsupset \bot, W_{1} \sqsubseteq \top, q_{1}(\ats{s_{1}},d\circ Alpha, W_{1}, Beta),$ \\
    & & $\vdots$ \\
    & & $W_{k} \sqsupset \bot, W_{k} \sqsubseteq \top, q_{k}(\ats{s_{k}},d\circ Alpha, W_{k}, Beta),$ \\
    & & $W = d\circ \bigsqcap\{W_{1}, \ldots, W_{k}\}$
\end{tabular}
\end{center}

\noindent The conditions in the body  of $C^{t}$ do indeed correspond to the performance of a $\sld{\qdomm}$ resolution step with clause $C$. In fact, $d\circ Alpha \sqsupseteq Beta$ checks that $C$ is eligible for such a step; the conditions in the next $k$ lines using new variables $W_i$ correspond to placing the annotated atoms from $C$'s body into the new goal; and the last condition introduces the proper defining constraint for $W$.

The idea for translating goals is similar. Given an initial goal  $\qlp{\qdomm}$ goal $G$ like \[q_{1}(\ats{t_{1}})\, \sharp\, W_{1}, \ldots, q_{m}(\ats{t_{m}})\, \sharp\,  W_{m} \sep W_{1} \sqsupseteq \beta_{1}, \ldots, W_{m} \sqsupseteq \beta_{m}\] where $\beta_{1}, \ldots, \beta_{m} \in \aqdomm$,  the translated goal $G^{t}$ is
 \[q_{1}(\ats{t_{1}},\top,W_1,\beta_1), \ldots, q_{m}(\ats{t_{m}},\top,W_m,\beta_m)\] where the three additional arguments at each atom are used to encode the initial threshold constraints $W_{i} \sqsupseteq \beta_{i}$, that are equivalent to $\top \circ W_{i} \sqsupseteq \beta_{i}$.

\begin{example}\label{exmpImplementation}
As an example of the translation process we present the translation of the program $\UProg$  from Example \ref{exmpProg} into a $\toy$ program which uses $\rdom$ constraints.
    \begin{verbatim}
  min1 [] = 1
  min1 [X|Xs] = min2 X (min1 Xs)
  min2 W1 W2 = if W1 <= W2 then W1 else W2
  data being = adam | eve | bird | cat | oak | apple
               | father being | mother being
  cruel(X,F,W,M) :- F*0.9>=M, W1>0, W1<=1.0, human(X,F*0.9,W1,M),
                    W2>0, W2<=1.0, eats(X,Y,F*0.9,W2,M),
                    W3>0, W3<=1.0, animal(Y,F*0.9,W3,M),
                    W == 0.9 * min1 [W1,W2,W3]
  cruel(X,F,W,M) :- F*0.4>=M, W1>0, W1<=1.0, human(X,F*0.4,W1,M),
                    W2>0, W2<=1.0, eats(X,Y,F*0.4,W2,M),
                    W3>0, W3<=1.0, plant(Y,F*0.4,W3,M),
                    W == 0.4 * min1 [W1,W2,W3]
  animal(bird,F,W,M) :- F*1.0>=M, W == 1.0 * min1 []
  animal(cat,F,W,M)  :- F*1.0>=M, W == 1.0 * min1 []
  plant(oak,F,W,M)   :- F*1.0>=M, W == 1.0 * min1 []
  plant(apple,F,W,M) :- F*1.0>=M, W == 1.0 * min1 []
  human(adam,F,W,M)      :- F*1.0>=M, W == 1.0 * min1 []
  human(eve,F,W,M)       :- F*1.0>=M, W == 1.0 * min1 []
  human(father(X),F,W,M) :- F*0.9>=M, W1>0, W1<=1.0,
        human(X,F*0.9, W1, M), W == 0.9 * min1 [W1]
  human(mother(X),F,W,M) :- F*0.8>=M, W1>0, W1<=1.0,
        human(X,F*0.8, W1, M), W == 0.8 * min1 [W1]
  eats(adam,X,F,W,M) :- F*0.8>=M, W == 0.8 * min1 []
  eats(eve,X,F,W,M) :- F*0.3>=M, W1>0, W1<=1.0,
        animal(X,F*0.3,W1,M), W == 0.3 * min1 [W1]
  eats(eve,X,F,W,M) :- F*0.6>=M, W1>0, W1<=1.0,
        plant(X,F*0.6,W1,M), W == 0.6 * min1 [W1]
  eats(father(X),Y,F,W,M) :- F*0.8>=M, W1>0, W1<=1.0,
        eats(X,Y,F*0.8,W1,M), W == 0.8 * min1 [W1]
  eats(mother(X),Y,F,W,M) :- F*0.7>=M, W1>0, W1<=1.0,
        eats(X,Y,F*0.7,W1,M), W == 0.7 * min1 [W1]
    \end{verbatim}

    To understand this example it is important to notice the following:
    \begin{enumerate}
        \item Since $glb$s in $\U$ are computed as minimums, translated programs must include functions for this task. Here, \texttt{min1} resp. \texttt{min2} compute the minimum of a list of numbers resp. two numbers.
        \item As $\toy$ need types for every constructor, we must include  suitable datatype declarations in translated programs.
        \item The resulting code could be simplified and optimized,  but our aim here is to illustrate the literal application of the general translation rules. For this reason, no optimizations have been performed. \qed
         \end{enumerate}
\end{example} 

%% file: TR_SIC-1-08_6-Conclusions.tex
\section{Conclusions and Future Work} \label{Conclusions}

We have generalized the early $QLP$ proposal by van Emden \cite{VE86} to a generic scheme $\qlp{\qdomm}$ parameterized by a qualification domain $\qdomm$, which must be a lattice with extreme points and equipped with an attenuation operator. The values belonging to a qualification domain are intended to qualify logical assertions, ensuring that they satisfy certain user's expectations. Qualification domains include $\B$ (classical truth values of two-valued logic), $\U$ (van Emden's certainty values) and $\W$ (numeric values representing proof weights), as well as arbitrary cartesian products of given qualification domains. As shown by  instances such as  $\qlp{\W}$ and $\qlp{\U \times \W}$, the $\qlp{\qdomm}$ scheme can express uncertainty in Logic Programming and more, since the user's expectations qualified by $\W$ do not correspond to uncertain truth values.

The semantic results obtained for $\qlp{\qdomm}$ are stronger than those in \cite{VE86}. Each program $\Prog$ has a least open Herbrand model $\Mp$ with two equivalent characterizations: as the least fixpoint of the operator  $\Tp$, and as the set of qualified atoms deducible from $\Prog$ in the logic calculus  $\qhl{\qdomm}$. Moreover, the goal solving calculus $\sld{\qdomm}$, based on an extension of $SLD$ resolution with qualification constraints, is sound and strongly complete for arbitrary open goals. $\sld{\B}$ boils down to classical $SLD$ resolution.

As implementation technique, we have proposed a translation of $\qlp{\qdomm}$ programs and goals into $\clp{\cdom{\qdomm}}$, choosing a constraint domain $\cdom{\qdomm}$ able to compute with qualification constraints over $\qdomm$. If $\qdomm$ is  $\U$, $\B$, or $\U \times \B$, the constraint domain $\cdom{\qdomm}$ can be chosen as $\rdom$, and $\qlp{\qdomm}$ can be implemented on top of any $CLP$ or $CFLP$ system which supports constraint solving over $\rdom$. We have implemented prototypes of  $\qlp{\U}$,  $\qlp{\W}$ and $\qlp{\U \times \W}$ on top of the $CFLP$ system $\toy$.

In comparison to the theory of generalized annotated logic programs ($GAP$ for short)  presented in  \cite{KS92}, our results in this report also include some interesting contributions. With respect to the syntax and goal solving procedure, the $\qlp{\qdomm}$ scheme can be made to fit into the  $GAP$ framework by viewing our attenuation operators as annotation functions. However, our resolution procedure $\sld{\qdomm}$ can be implemented more efficiently than the constrained $SLD$ resolution used in $GAP$, due to an optimized treatment of qualification constraints and, more importantly, because the costly computation of so-called {\em reductants} between variants of program clauses is needed in $GAP$ resolution but not in $\sld{\qdomm}$. The purpose of reductants in $GAP$ is to explicitly compute the $lub$s of several lattice values (qualification values in the case of $\qlp{\qdomm}$) which would result from finitely many different computations if no reductants were used. In $GAP$'s declarative semantics, interpretations are required to be closed w.r.t. finite $lub$s of lattice values assigned to the same atom, and for this reason  reductants are needed for the completeness of goal resolution. In $\qlp{\qdomm}$ interpretations as defined in Section \ref{Language} no closure condition w.r.t. $lub$s is required, and therefore the completeness result stated in Theorem \ref{thmCompleteness} can be proved without reductants. Of course, the $\qlp{\qdomm}$ approach to semantics means that a user has to observe \emph{several} computed answers for one and the same goal and think of the \emph{lub} of the various $\qdomm$ elements provided by the different computations by himself instead of getting the \emph{lub} computed by one single $\sld{\qdomm}$ derivation. In our opinion, this is a reasonable scenario because even in $GAP$ the $\mathcal{T}$ value provided by any single computed answer always corresponds to some $lub$ of finitely many  $\mathcal{T}$ values,  and it may be not the highest possible $\mathcal{T}$ value w.r.t. to the program's declarative semantics. Moreover, our Theorem  \ref{thmCompleteness} is much stronger than the one given in \cite{KS92}, which only ensures the possibility of computing {\em some} solution for any goal whose solvability holds in the least program model. We strongly conjecture that a stronger completeness theorem could be proved also for $GAP$ by using a proof technique more similar to our's.

As possible lines of future work we consider: to improve the current prototype implementations of the instances $\qlp{\U}$, $\qlp{\W}$ and $\qlp{\U \times \W}$; to extend the $\qlp{\qdomm}$ scheme and its implementation to a more expressive scheme which can support qualitative programming with features such as disjunctive goals, negation, lazy functions and parametrically given constraint domains; to explore alternative semantic approaches, considering annotations, bilattices, probabilistic semantics and similarity based unification; and to investigate applications to the computation of qualified answers for web search queries.

\section*{Acknowledgements}
The authors are thankful to their colleagues Paco L\'opez and Rafa Caballero for their valuable hints concerning bibliography and implementation techniques. They also appreciate the constructive comments of the anonymous reviewers, that were helpful for improving the presentation.

%% file: TR_SIC-1-08_7-Appendix.tex
\section{Appendix: Additional Proofs} \label{Proofs}
This appendix contains the proofs of the  Soundness Theorem \ref{thmSoundness} and 
the Completeness Theorem \ref{thmCompleteness}. In order to prove them, we will 
previously prove an auxiliary lemma for each of the two cases.
In the rest of the Appendix, we assume a given program $\Prog$ over a qualification domain $\qdomm$.

\subsection{Proof of the Soundness Theorem}

\begin{lemma}[Soundness]\label{lemmaSoundness}
Assume two goals $G_{0}$ and  $G_{1}$  and a pair of substitutions $(\theta, \rho)$ 
such that $G_{0} \qres{C_{1},\sigma_{1}} G_{1}$ and $(\theta,\rho) \in \qsol{G_{1}}$
(the set of all solutions of $G_{1}$). Then we have that $(\theta, \rho) \in \qsol{G_{0}}$.
\end{lemma}
\begin{proof}
Assume $G_{0} \equiv \ats{L},\at{A}{W},\ats{R}\sep\sigma_{0}\sep \alpha\circ W \sqsupseteq \beta, \Delta$; $C_{1} \equiv (H \qgets{d} B_{1}, \ldots, B_{k}) \in_{\mathrm{var}} \Prog$ a variant of a program clause without variables in common with $G_{0}$; and $\sigma_{1}$ the m.g.u. between $A$ and $H$. Then \[G_{1} \equiv (\ats{L}, \at{B_{1}}{W_{1}}, \ldots, \at{B_{k}}{W_{k}}\ats{R})\sigma_{1} \sep \sigma_{0}\sigma_{1} \sep \Delta_{1}\] where $\Delta_{1} \equiv W = d \circ \infi\{W_{1}, \ldots, W_{k}\}, d\circ\alpha\circ W_{1} \sqsupseteq \beta, \ldots, d\circ\alpha\circ W_{k} \sqsupseteq \beta, \Delta$. As $(\theta,\rho) \in \qsol{G_{1}}$ we know
\begin{itemize}
    \item[(1)] $\sigma_{0}\sigma_{1}\theta = \theta$ ,
    \item[(2)] $\rho \in Sol(\Delta_{1})$ , and
    \item[(3)] $\Prog \qhld (\ats{L},\at{B_{1}}{W_{1}}, \ldots, \at{B_{k}}{W_{k}}, \ats{R})\sigma_{1}\app(\theta,\rho)$\footnote{$\Prog \qhld (\at{A}{W})\app(\theta,\rho) \diff \Prog \qhld \at{A\theta}{W\rho} \enspace .$} .
\end{itemize}
And for $(\theta,\rho)$ to be a solution of $G_{0}$ we need the following:
\begin{itemize}
    \item[(4)] $\sigma_{0}\theta = \theta$ ,
    \item[(5)] $\rho \in Sol(\alpha\circ W \sqsupseteq \beta, \Delta)$ , and
    \item[(6)] $\Prog \qhld (\ats{L},\at{A}{W},\ats{R})\app(\theta,\rho)$ .
\end{itemize}
Therefore we have to prove (4), (5) and (6).

\paragraph{Proof of (4).} First, we can see that for every variable $y \in \Var$, if $y \in \mbox{vran}(\sigma_{0})$, then $y \notin \mbox{dom}(\sigma_{0})$ because $\sigma_{0}$ is idempotent. Hence, $y \in \mbox{vran}(\sigma_{0}) \Longrightarrow y\sigma_{1}\theta = y\sigma_{0}\sigma_{1}\theta =_{\mathrm{(1)}} y\theta$, and it is true that (7) $\sigma_{1}\theta = \theta$ $[\mbox{vran}(\sigma_{0})]$. Now, for any variable $x$ we can prove
$x\sigma_0\theta = x\theta$ by distinguishing two cases: 
a) if $x \notin \mbox{dom}(\sigma_{0})$ then $x\sigma_{0}\theta = x\theta$;
 and b) if $x \in \mbox{dom}(\sigma_{0})$ then $\mbox{var}(x\sigma_{0}) \subseteq \mbox{vran}(\sigma_{0}) \Longrightarrow x\sigma_{0}\theta =_{\mathrm{(7)}} x\sigma_{0}\sigma_{1}\theta =_{\mathrm{(1)}} x\theta$.

\paragraph{Proof of (5).} We have to prove that $\alpha\circ W\rho \sqsupseteq \beta$ and $\rho \in Sol(\Delta)$. $\alpha\circ W\rho =_{\mathrm{(2)}} \alpha\circ d\circ\infi\{W_{1}\rho, \ldots, W_{k}\rho\} = \infi\{\alpha\circ d\circ W_{1}\rho, \ldots, \alpha\circ d\circ W_{k}\rho\}$. 
It is enough proving $\alpha\circ d \circ W_{i}\rho \sqsupseteq \beta$ for $1 \le i \le k$. But $\{\alpha\circ d \circ W_{i} \sqsupseteq \beta \mid 1 \le i \le k\} \subseteq \Delta_{1}$ and $\rho \in Sol(\Delta_{1})$. $\rho \in Sol(\Delta)$ is trivial because $\Delta \subseteq \Delta_{1}$.

\paragraph{Proof of (6).} We can split (6) in the following three cases:
\begin{itemize}
    \item[(6a)] $\Prog \qhld \ats{L} \app (\theta,\rho)$. We prove that $\ats{L}\app(\theta,\rho) = \ats{L}\sigma_{1}\app(\theta,\rho)$ which, because of (3), can be inferred in $\qhl{\qdomm}$ from $\Prog$. We know that $\mbox{dom}(\sigma_{0}) \cap \mbox{var}(G_{0}) = \emptyset$, therefore, $\ats{L}\sigma_{1}\app(\theta,\rho) = \ats{L}\sigma_{0}\sigma_{1}\app(\theta,\rho) = \ats{L}\app(\sigma_{0}\sigma_{1}\theta,\rho) =_{(1)} \ats{L}\app(\theta,\rho)$.
    \item[(6b)] $\Prog \qhld \at{A}{W}\app(\theta,\rho)$. Using $W\rho = d\circ \infi\{W_{1}\rho, \ldots, W_{k}\rho\}$ which holds because of (2), 
    (3) and one inference step with clause $C_{1}$ and substitution $\sigma_{1}\theta$, we obtain $\Prog \qhld \at{H\sigma_{1}\theta}{W\rho}$.
    Now, because $\sigma_{1}$ is the m.g.u. between $A$ and $H$, we have $H\sigma_{1}\theta = A\sigma_1\theta$. 
    Therefore, we have $\Prog \qhld \at{A\sigma_{1}\theta}{W\rho}$. Finally, we note that $A\sigma_{1}\theta = A\sigma_0\sigma_1\theta =_{(1)} A\theta$ 
    because $\mbox{dom}(\sigma_0) \cap \mbox{var}(A) = \emptyset$.
       \item[(6c)] $\Prog \qhld \ats{R}\app(\theta,\rho)$. As in (6a). \qed
\end{itemize}
 \end{proof}

\subsubsection{Proof (of Soundness Theorem).}
Assume $G_{0} \qresn{n}{\sigma'} G$ where $G \equiv \sigma \sep \Delta$ is solved. Let $(\sigma,\mu)$ be the solution associated to $G$. We prove $(\sigma, \mu) \in \qsol{G_{0}}$ by  induction on $n$.

\begin{description}
            \item[Base.] In  this case, $n = 0$ and $G_{0} = G$ is solved and $\Prog \qhld \ats{A} \app(\sigma,\mu)$ is trivial because the sequence of atoms $\ats{A}$ of $G$  is empty. Moreover, $\mu \in Sol(\Delta)$ because $\mu = \omega_{\Delta}$.
            \item[Induction.] In this case we have $n > 0$ and $G_{0} \qres{} G_{1} \qresn{n-1}{} G$. Then we obtain $(\sigma, \mu) \in \qsol{G_{1}}$ by induction hypothesis, and therefore $(\sigma, \mu) \in \qsol{G_{0}}$ because of Lemma  \ref{lemmaSoundness}. \qed
 \end{description}

\subsection{Proof of the Completeness Theorem}

Before going into the proof, just a note on notation: as said, $(\theta, \rho) \in \qsol{G}$ means that the pair of substitutions $(\theta, \rho)$ is a solution of the goal $G \equiv \ats{A}\sep\sigma\sep\Delta$. Now, writing $(\theta, \rho) \in \qsoln{n}{G}$ we are expressing that the exact number of inference steps in $\Prog \qhld \ats{A}\app(\theta,\rho)$ is $n$, written as $\Prog \qhldn{n} \ats{A}\app(\theta,\rho)$.

\begin{lemma}[Completeness]\label{lemmaCompleteness}
Let $G_{0} \equiv \ats{A_{0}}\sep\sigma_{0}\sep\Delta_{0}$ be a goal not solved, and $(\theta_{0},\rho_{0}) \in \qsoln{n}{G_{0}}$. Let also $V_{0}$ be any finite set of variables such that $\mathrm{var}(G_{0}) \cup \mathrm{dom}(\theta_{0}) \subseteq V_{0}$. For any arbitrary selection of an atom $\at{A}{W}$ of $\ats{A_{0}}$, there exists some resolution step $G_{0} \qres{\sigma_{1}} G_{1}$ selecting the chosen atom and, in addition, some $(\theta_{1},\rho_{1})$ satisfying the following properties:
\begin{enumerate}
    \item[a.] $\theta_{1} = \theta_{0} \, [V_{0}]$
    \item[b.] $\sigma_{1}\theta_{1} = \theta_{1}$
    \item[c.] $\sigma_{0}\sigma_{1}\theta_{1} = \theta_{1}$
    \item[d.] $\rho_{1} \sqsupseteq \rho_{0} \, [\mathrm{war}(G_{0})]$
    \item[e.] $\rho_{1} \in Sol(\Delta_{1})$
    \item[f.] $\Prog \qhldn{n-1} \ats{A_{1}}\app(\theta_{1},\rho_{1})$
\end{enumerate}
In particular, (c), (e) and (f) mean that $(\theta_{1},\rho_{1}) \in \qsoln{n-1}{G_{1}}$.
\end{lemma}
\begin{proof}
Assume $\at{A}{W}$ to be the selected atom in $G_{0}$. Then, $G_{0} \equiv \ats{L_{0}},\at{A}{W},\ats{R_{0}}$ $\sep\sigma_{0}\sep \alpha\circ W \sqsupseteq \beta, \Delta$. Because of the lemma's hypothesis we can also assume the following:
\begin{itemize}
    \item[(0)] $\sigma_{0}\theta_{0} = \theta_{0}$ ,
    \item[(1)] $\rho_{0} \in Sol(\Delta_{0})$ ,
    \item[(2)] $\Prog \qhldn{m_{1}} \ats{L_{0}}\app(\theta_{0},\rho_{0})$ ,
    \item[(3)] $\Prog \qhldn{m_{2}} \at{A}{W}\app(\theta_{0},\rho_{0})$ , and
    \item[(4)] $\Prog \qhldn{m_{3}} \ats{R_{0}}\app(\theta_{0},\rho_{0})$ 
\end{itemize}
with $m_{1} + m_{2} + m_{3} = n > 0$.

Because of (3) there must exist some clause $C_{1} \equiv (H \qgets{d} B_{1}, \ldots, B_{k}) \in_{\mathrm{var}} \Prog$ and some substitution $\eta_{0}$ such that 
\begin{itemize}
    \item[(5)] $A\theta_{0} = H\eta_{0}$ and $\Prog \qhldn{m_{2}-1} \at{B_{1}\eta_{0}}{d_{1}}, \ldots, \at{B_{k}\eta_{0}}{d_{k}}$ with $d_{1}, \ldots, d_{k} \in \aqdomm$ such that $W\rho_{0} \sqsubseteq d\circ \infi\{d_{1}, \ldots, d_{k}\}$.
\end{itemize}
It is possible to choose $C_{1}$ and $\eta_{0}$ so that $\mathrm{var}(C_{1}) \cap V_{0} = \emptyset$ and $\mathrm{dom}(\eta_{0}) \subseteq \mathrm{var}(C_{1})$. Therefore, it is guaranteed that $\mathrm{dom}(\eta_{0}) \cap \mathrm{dom}(\theta_{0}) = \emptyset$ and then:
\begin{itemize}
    \item[(6)] $\theta_{1} =_{\mathrm{def}} \theta_{0} \uplus \eta_{0}$ is a well-founded substitution that satisfies: $\mathrm{dom}(\theta_{1}) = \mathrm{dom}(\theta_{0}) \uplus \mathrm{dom}(\eta_{0})$; $\theta_{1} = \theta_{0} \, [V_{0}]$; $\theta_{1} = \eta_{0} \, [\backslash V_{0}] \Longrightarrow$ ($a$) of lemma.
\end{itemize}

From (5) and (6) we know that $\theta_{1}$ is an unifier of $A$ and $H$. Choosing $\sigma_{1}$ as the m.g.u. (in the Robinson's sense) between $A$ and $H$ we will have:
\begin{itemize}
    \item[(7)] $A\sigma_{1} = H\sigma_{1}$ and $\sigma_{1}\theta_{1} = \theta_{1} \Longrightarrow$ ($b$) of lemma.
\end{itemize}
Then, taking $\rho_{1}$ such that
\begin{itemize}
    \item[(8)] $W'\rho_{1} =_{\mathrm{def}}
        \left\{\begin{array}{ll}
            d_{i} & \textrm{if } W' = W_{i} \textrm{ for some } 1 \le i \le k \\
            d\circ \infi\{d_{1}, \ldots, d_{k}\} & \textrm{if } W' = W \\
            W'\rho_{0} & \textrm{otherwise}
        \end{array}\right.$
\end{itemize}
we will have, by (8) and (5), $\rho_{1} \sqsupseteq \rho_{0} \, [\mathrm{war}(G_{0})] \Longrightarrow$ ($d$) of lemma. 

Now, doing a resolution step with $\sigma_{1}$ and $C_{1}$ we get $G_{0} \qres{\sigma_{1},C_{1}} G_{1} \equiv \ats{A_{1}}\sep\sigma_{0}\sigma_{1}\sep \Delta_{1}$ where $\ats{A_{1}} = (\ats{L_{0}},\at{B_{1}}{W_{1}},\ldots, \at{B_{k}}{W_{k}},\ats{R_{0}})\sigma_{1}$ and $\Delta_{1} \equiv d\circ \alpha\circ W_{1} \sqsupseteq \beta, \ldots, d\circ \alpha\circ W_{k} \sqsupseteq \beta, W = d\circ \infi\{W_{1}, \ldots, W_{k}\}, \Delta$. Note that we can deduce $d\circ\alpha\sqsupseteq\beta$ from (5), (1) and the axioms required for the attenuation operation $(\circ)$ in any qualification domain, because we have $W\rho_{0} = d\circ\infi\{d_{1},\ldots,d_{k}\}$ and $\alpha\circ W\rho_{0} \sqsupseteq \beta \Longrightarrow \alpha\circ d\circ\infi\{d_{1},\ldots,d_{k}\} \sqsupseteq \beta \Longrightarrow \alpha\circ d\sqsupseteq \beta$.
Remember from Definition \ref{defStep} that the condition $\alpha\circ d\sqsupseteq \beta$ is required for the resolution step to be enabled.

To finish the proof we only need to prove ($c$), ($e$) and ($f$).

\paragraph{Proof of (c).}
$\sigma_{0}\underline{\sigma_{1}\theta_{1}} =_{(b)} \sigma_{0}\theta_{1} =_{(6)} \sigma_{0}(\theta_{0}\uplus\eta_{0}) =_{(*)} \sigma_{0}\theta_{0}\uplus\eta_{0} =_{(0)} \theta_{0}\uplus\eta_{0} = \theta_{1}$. \\
($*$) Because $\mathrm{vran}(\sigma_{0}) \subseteq V_{0}$ and $\mathrm{dom}(\eta_{0}) \cap V_{0} = \emptyset$.

\paragraph{Proof of (e).}
We have to see that $\rho_{1}$ satisfies every constraint in $\Delta_{1}$:
\begin{enumerate}
    \item $W = d\circ \infi\{W_{1}, \ldots, W_{k}\}$. This is satisfied by definition of $\rho_{1}$.
    \item $d\circ \alpha\circ W_{i} \sqsupseteq \beta$ for $1 \le i \le k$. We know from (1) that $\alpha\circ W\rho_{0} \sqsupseteq \beta$, and from (d) follows $W\rho_{1} \sqsupseteq W\rho_{0}$. Therefore, $\alpha\circ W\rho_{1} \sqsupseteq \beta$.
    Because of (8) we also know $W\rho_{1} = d\circ \infi\{W_{1}\rho_{1}, \ldots, W_{k}\rho_{1}\}$ that implies
    $\alpha \circ W\rho_{1}  =  \infi\{d \circ \alpha \circ W_{1}\rho_{1}, \ldots, d \circ \alpha \circ W_{k}\rho_{1}\}$
    Hence $\alpha\circ W\rho_{1}  \sqsupseteq \beta$ implies that 
    $d\circ\alpha\circ W_{i}\rho_{1} \sqsupseteq \beta$ is satisfied for every $1 \le i \le k$.
    \item $\Delta$. From (7) follows that $\rho_{1} = \rho_{0} \, [\mathrm{war}(\Delta)]$, and because of (1) $\rho_{0} \in Sol(\Delta) \Longrightarrow \rho_{1} \in Sol(\Delta)$.
\end{enumerate}

\paragraph{Proof of (f).}
First we can deduce:
\begin{itemize}
    \item[(9)] $\Prog \qhldn{m_{1}} \ats{L_{0}}\app(\theta_{1},\rho_{1})$ by (2), (6) and (8).
    \item[(10)] $\Prog \qhldn{m_{2}-1} (\at{B_{1}}{W_{1}}, \ldots, \at{B_{k}}{W_{k}})\app(\theta_{1},\rho_{1})$ by (5), (6) and (8).
    \item[(11)] $\Prog \qhldn{m_{3}} \ats{R_{0}}\app(\theta_{1},\rho_{1})$ by (4), (6) and (8).
\end{itemize}
Considering that $m_{1} + (m_{2} - 1) + m_{3} = n - 1$, (9), (10) and (11) imply that $\Prog \qhldn{n-1} (\ats{L_{0}},\at{B_{1}}{W_{1}}, \ldots, \at{B_{k}}{W_{k}},\ats{R_{0}})\app(\theta_{1},\rho_{1})$; and this is ($f$) due to ($b$). \qed

\end{proof}

\subsubsection{Proof (of Completeness Theorem).}
As $G_{0}$ is a goal, we know that $\sigma_{0}$ is idempotent and that $\ats{A_{0}}\sigma_{0} = \ats{A_{0}}$. Now, as $(\theta_{0},\rho_{0}) \in \qsol{G_{0}}$, we can choose a number $n\in\NAT$ such that $(\theta_{0},\rho_{0}) \in \qsoln{n}{G_{0}}$ and therefore we have (1) $\sigma_{0}\theta_{0} = \theta_{0}$, (2) $\rho_{0} \in Sol(\Delta_{0})$ and (3) $\Prog \qhldn{n} \ats{A_{0}} \app (\theta_{0},\rho_{0})$. We can also choose a finite set of variables $V_{0}$ satisfying (4) $\mathrm{var}(G_{0}) \cup \mathrm{dom}(\theta_{0}) \subseteq V_{0}$. Given the conditions (1) to (4), we will prove the following:
\begin{itemize}
    \item[($\dagger$)] There exist some resolution computation $G_{0} \qresn{*}{\sigma} \sigma_{0}\sigma\sep\Delta$ (that we can build making use of any selection strategy) ending in a solved goal, and some substitution $\theta$ satisfying (5) $\theta = \theta_{0} \, [V_{0}]$, (6) $\sigma\theta = \theta$ and (7) $\sigma_{0}\sigma\theta = \theta$.
\end{itemize}
From ($\dagger$) follows the theorem's thesis (except (8) $\mu \sqsupseteq \rho_{0} \, [\mathrm{war}(G_{0})]$) because:
\begin{itemize}
    \item $(6) \Longrightarrow_{(5)} \sigma\theta = \theta_{0} \, [V_{0}] \Longrightarrow \sigma \preccurlyeq \theta_{0} \, [\mathrm{var}(G_{0})]$
    \item $(7) \Longrightarrow_{(5)} \sigma_{0}\sigma\theta = \theta_{0} \, [V_{0}] \Longrightarrow \sigma_{0}\sigma \preccurlyeq \theta_{0} \, [\mathrm{var}(G_{0})]$
\end{itemize}
We simultaneously prove ($\dagger$) and (8) by induction on $n$:

\begin{description}
    \item[Base.] If $n=0$, (2) implies that $\ats{A_{0}}$ is empty. Then, taking $\sigma = \epsilon$ and $\theta = \theta_{0}$ we have that $G_{0} \qresn{0}{\epsilon} \sigma_{0}\sep\Delta_{0}$ with is a trivial resolution of 0 steps; and in addition:
    \begin{itemize}
        \item (5) reduces to $\theta_{0} = \theta_{0} \, [V_{0}]$, which is trivial.
        \item (6) reduces to $\theta_{0} = \theta_{0}$, which also is trivial.
        \item (7) reduces to $\sigma_{0}\theta_{0} = \theta_{0}$ which is true given (1).
        \item (8) is satisfied because $\mu = \rho_{0} \, [\mathrm{war}(G_{0})]$ is true, given that $\mathrm{war}(G_{0}) = \mathrm{war}(\Delta_{0})$. Now, as $\Delta_{0}$ is solved, it only contains defining constraints and therefore (2) implies that for any $W \in \mathrm{war}(\Delta_{0})$ it is true that $W\rho_{0} = \omega_{\Delta_{0}}(W) = W\mu$.
    \end{itemize}

    \item[Induction.] If $n>0$, (2) implies that $\ats{A_{0}}$ is not empty. Hence, selecting an atom $\at{A}{W}$ in $\ats{A_{0}}$ with any selection strategy and using the Completeness Lemma \ref{lemmaCompleteness}, we can perform a resolution step \[(9) \quad G_{0} \equiv \ats{A_{0}}\sep\sigma_{0}\sep\Delta_{0} \qres{\sigma_{1}} \ats{A_{1}}\sep\sigma_{0}\sigma_{1}\sep\Delta_{1} \equiv G_{1}\] having that there exists some solution $(\theta_{1},\rho_{1}) \in \qsoln{n-1}{G_{1}}$ satisfying all 6 conditions guaranteed by the lemma: (10) $\theta_{1} = \theta_{0} \, [V_{0}]$, (11) $\sigma_{1}\theta_{1} = \theta_{1}$, (1') $\sigma_{0}\sigma_{1}\theta_{1} = \theta_{1}$, (12) $\rho_{1} \sqsupseteq \rho_{0} \, [\mathrm{war}(G_{0})]$, (2') $\rho_{1} \in Sol(\Delta_{1})$, and (3') $\Prog \qhldn{n-1} \ats{A_{1}}\app(\theta_{1},\rho_{1})$.

    Let $V_{1}$ be any finite set of variables such that \[(4') \quad V_{0} \cup \mathrm{var}(G_{1}) \cup \mathrm{dom}(\theta_{1}) \subseteq V_{1} \enspace .\]

    Conditions (1'), (2'), (3') and (4') are similar, respectively, to (1), (2), (3) and (4), but now for $(\theta_{1}, \rho_{1}) \in \qsoln{n-1}{G_{1}}$. By 
    induction hypothesis we can obtain a resolution computation \[(13) \quad G_{1} \qresn{*}{\sigma'} \sigma_{0}\sigma_{1}\sigma'\sep\Delta'\] and a substitution $\theta$ such that
    \begin{itemize}
        \item[] (5') $\theta = \theta_{1} \, [V_{1}] \quad$ (6') $\sigma'\theta = \theta \quad$ (7') $\sigma_{0}\sigma_{1}\sigma'\theta = \theta$
        \item[] (8') $\mu' \sqsupseteq \rho_{1} [\mathrm{war}(G_{1})]$ with $(\sigma_{0} \sigma_{1} \sigma', \mu')$ the associated solution to $\sigma_{0}\sigma_{1}\sigma'$ $\sep \Delta'$.
    \end{itemize}

    From (9) and (13) results \[G_{0} \equiv \ats{A_{0}}\sep\sigma_{0}\sep\Delta_{0} \qres{\sigma_{1}} G_{1} \equiv \ats{A_{1}}\sep\sigma_{0}\sigma_{1}\sep\Delta_{1} \qresn{*}{\sigma'} \sigma_{0}\sigma_{1}\sigma'\sep\Delta' \enspace .\]

    Now there is only left to check that (5), (6) and (7) are satisfied given the same $\theta$ that satisfies (5'), (6') and (7') and $\sigma = \sigma_{1}\sigma'$; and that (8) is also satisfied when $\mu = \mu'$. In fact:
    \begin{itemize}
        \item (5) trivially follows from (5'), (4') and (10).
        \item (6) comes from the following: by (4') we can assume $\theta = \theta_{1} \uplus \eta'$, with $\eta'$ such that $\mathrm{dom}(\eta') \cap (V_{1}) = \emptyset$. Then: $\underline{\sigma}\theta = \sigma_{1}\underline{\sigma'\theta} =_{(6')} \sigma_{1}\underline{\theta} = \sigma_{1}(\theta_{1} \uplus \eta') =_{(*)} \underline{\sigma_{1}\theta_{1}} \uplus \eta' =_{(11)} \theta_{1} \uplus \eta' = \theta$.
The step $(*)$ is correct because $\mathrm{vran}(\sigma_{1}) \subseteq V_{1}$ and $\mathrm{dom}(\eta') \cap V_{1} = \emptyset$.
        \item (7) trivially follows from (7'), given that $\sigma = \sigma_{1}\sigma'$.
        \item (8) is consequence of (8') and (12), because $\mathrm{war}(G_{0}) \subseteq \mathrm{war}(G_{1})$. \qed
    \end{itemize}
\end{description} 